\documentclass[letterpaper, 10 pt, conference]{ieeeconf_ecc19}  
\IEEEoverridecommandlockouts                
\usepackage{graphics} 
\usepackage{epsfig} 
\usepackage{mathptmx}
\usepackage{times} 
\usepackage{amsmath} 
\usepackage{amssymb}  
\usepackage{xcolor}
\usepackage{hyperref}

\newtheorem{theorem}{Theorem}

\newtheorem{definition}{Definition}

\title{\LARGE \bf
Dissipative iFIR filters for data-driven design
}

\author{Zixing Wang$^{a}$,  Yi Zhang$^{a}$, and  Fulvio Forni$^{a}$
\thanks{*The work of Zixing Wang was supported by CSC Cambridge Scholarship.  The work of Yi Zhang was supported by Sciences Research Council and AgriFoRwArdS CDT EP/S023917/1.}
\thanks{$^{a}$ Department of Engineering, University of Cambridge, CB2 1PZ Cambridge, U.K.
         \{{\tt\small zxw20@.cam.ac.uk}, {\tt\small yz892@cam.ac.uk}, {\tt\small f.forni@eng.cam.ac.uk}\}.}
}

\begin{document}

\maketitle
\thispagestyle{empty}
\pagestyle{empty}

\begin{abstract}

We tackle the problem of providing closed-loop stability guarantees with a scalable data-driven design. We combine virtual reference feedback tuning with dissipativity constraints on the controller for closed-loop stability. The constraints are formulated as a set of linear inequalities in the frequency domain. This leads to a convex problem that is scalable with respect to the length of the data and the complexity of the controller.
An extension of virtual reference feedback tuning to include disturbance dynamics is also discussed. The proposed data-driven control design is illustrated by a soft gripper impedance control example.
\end{abstract}

\section{INTRODUCTION}
One of the central research problems of data-driven control is ensuring closed-loop stability. To design stabilizing controllers, most approaches combine a data-based representation of the system with optimization tools, such as linear matrix inequalities (LMI) or sums of squares (SOS). For linear time-invariant systems, representative examples include the behavioral approach of \cite{bib:Tesi2020} and informativity \cite{bib:Henk2023}. In the nonlinear setting, examples include approaches in the setting of parameter-varying systems \cite{bib:Verhoek2024b} and polynomial interpolation \cite{bib:Guo2022a, bib:Martin2023b}. Although the use of LMI / SOS constraints provides an effective approach to address the stability of data-driven controllers in closed loop, solving LMI / SOS formulations may face scalability issues when the dimension of the problem increases. 

In this paper, we propose a data-driven control design method that guarantees \emph{closed-loop stability} and is \emph{scalable} with respect to the length of training data and the complexity of the controller. For \emph{closed-loop stability}, we follow the philosophy of dissipativity learning control \cite{bib:Tang2021}
\textcolor{black}{. We consider an uncertain family of plant characterized by a common dissipativity property. We assume this property is obtained through estimation tools such as \cite{bib:Koch2022,bib:Martin2023c}. We then take into account this property for control synthesis, to guarantee stability. We tackle the synthesis of the controller
} 
by applying virtual reference feedback tuning (VRFT) constrained to compatible dissipativity conditions on the controller. The resulting closed loop is guaranteed to be stable, according to passivity or small-gain theorems.

For \emph{scalability}, we take advantage of the particular structure of iFIR controllers (integrator plus finite impulse response filter) \cite{bib:Wang2024,bib:Formentin2011,bib:vanHeusden2011,bib:Yahagi2022-overfitting}, to derive dissipativity constraints that take the form of a finite number of linear inequalities in the frequency domain. We achieve this by first deriving a graphical representation of the dissipativity property of interest. Then, this `feasible' region is approximated using polygons and expressed in terms of a finite set of linear inequalities, through sampling. {\color{black}In contrast to other contributions
\cite[Sec. 4.2]{bib:Tang2021}, the resulting constraints do not depend on the length of training data and scale linearly with the order of the controller.}  
The proposed synthesis \textcolor{black}{also} extends \cite{bib:Wang2024} to more general dissipativity properties, including excess/shortage of passivity and gain conditions. These make the synthesis more flexible, reducing the conservativeness. With the proposed design, the goals of closed-loop stability and closed-loop performance are decoupled; the former is structurally guaranteed via dissipativity and thus not affected by
data scarcity or low-quality data. 

As a side contribution, we also extend the VRFT control synthesis to take into account disturbance dynamics, {\color{black}generalizing the case of additive disturbances at the input and output of the plant \cite{bib:Jeng2016, bib:Eckharda2018, bib:Campestrini2018, bib:Campi2002b}. The handling of disturbance dynamics is desirable in some applications, such as impedance control in mechanics \cite{bib:Hogan1985a,bib:Colgate1988,bib:Hogan2021}, to achieve a predefined target-compliant behavior beyond the rejection of external disturbance forces}.
 Compared to \cite{bib:Rojas2011, bib:Masuda2017}, our approach does not need any measure of the disturbance. It is based on a 2 degree-of-freedom controller \cite{bib:Campi2002b} {\color{black}to shape reference tracking and disturbance sensitivity} simultaneously. 

The paper is organized as follows. Section \ref{sec:vrft} revises the VRFT synthesis and extends it to disturbance dynamics. Section \ref{sec:dissipativityDefi} revises the classical dissipativity and
related graphical interpretation. Section \ref{sec:dissi_constraints} discusses in detail the dissipativity constraints that achieve closed-loop stability, based on LMIs and on constraints in the frequency domain. The proofs of the main theorems are given in Section \ref{sec:proof}. Section \ref{sec:examples} illustrates the effectiveness of our approach in an example focusing on impedance control for a soft gripper.

\section{Data-driven Design of iFIR Controllers} \label{sec:vrft}
An iFIR controller $C$ \cite{bib:Wang2024} of order $m_{\mathrm{fb}}\in \mathbb{N}$ is a parallel interconnection between an integrator and a finite impulse response (FIR) filter 
\begin{equation}  \label{eq:ifir}
    C(z) = \underbrace{\frac{\gamma T_{s}}{1-z^{-1}}}_{\mbox{integrator}} + \underbrace{\sum_{t=0}^{m_{\mathrm{fb}}-1} g^{\mathrm{fb}}(t)  z^{-t}}_{\mbox{FIR}}
\end{equation}
where $\mathbf{g}^{\mathrm{fb}} = \{g^{\mathrm{fb}}(t)  \}_{t=0}^{m_{\mathrm{fb}}-1}$ are the impulse response coefficients of the FIR filter.
In comparison to the classical PID controller, the iFIR controller keeps the integral action and
replaces the action of proportional and derivative components with a general FIR filter, to
provide more flexibility in shaping the closed-loop performance.

Consider the closed loop in Figure \ref{fig:blockdiag1DOF_P1}, where $P_{1}$ is a stable, discrete, possibly nonlinear, time-invariant, single-input-single-output (SISO) plant, whose dynamics are unknown. $C$ is the iFIR controller in \eqref{eq:ifir}. 
Given a stable, discrete reference model $M_{r}$ chosen by the user, we aim at tuning the iFIR controller $C$ so that the closed-loop response $\mathbf{y}= \{y(t)\}_{t=0}^{N-1}$, for an arbitrary reference $\mathbf{r}= \{r(t)\}_{t=0}^{N-1}$, matches the target response $\mathbf{y}^{*}= \{ M_{r}(r)(t)\}_{t=0}^{N-1}$. That is, for all $N \geq 0$, find an iFIR controller $C$ as the minimizer of $\min || \mathbf{y} - \mathbf{y}^{*}   ||_{2}^{2}$.

\begin{figure}[htbp]
    \begin{center}
        \includegraphics[width=0.66\columnwidth]{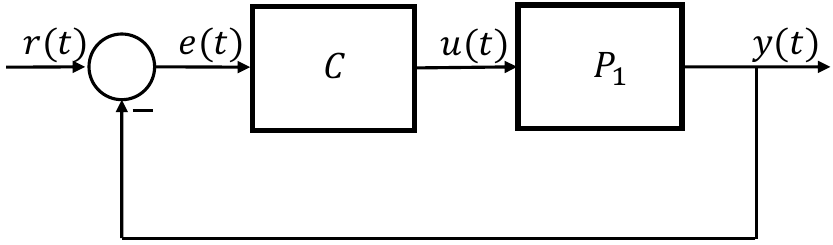} \vspace{-2mm}
        \includegraphics[width=0.3\columnwidth]{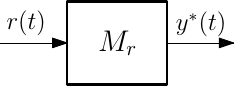}
        \caption{A block diagram of a closed-loop system and of its desired behavior, modeled by the reference model $M_r$.} 
        \label{fig:blockdiag1DOF_P1}
    \end{center}
\end{figure}
\vspace{-3mm}

We achieve this control goal using virtual reference feedback tuning (VRFT) \cite{bib:Campi2002}.
We solve the following data-fitting problem
\begin{equation} \label{eq:obj1DOF_P1}
    \min_{C(z)} \frac{1}{N} \sum_{t=0}^{N-1} \left(u(t) - C(M_{r}^{-1}(y) - y)(t) \right)^{2} 
\end{equation}
where $\mathbf{u} = \{u(t)\}_{t=0}^{N-1}$ is the exciting signal applied to $P_{1}$ and $\mathbf{y} = \{y(t)\}_{t=0}^{N-1}$ is the measured plant's output. The objective function in \eqref{eq:obj1DOF_P1} can be interpreted as follows. Assume $\mathbf{u}, \mathbf{y}$ are data measured in a closed loop with an ideal controller $C^{*}$ that achieves perfect matching $(\mathbf{y} = \mathbf{y}^*)$. We can compute the data of the \emph{virtual reference} $\mathbf{r} = \{r(t)\}_{t=0}^{N-1}$ as $r(t) = M_{r}^{-1}(y)(t)$. Hence, we obtain the input-output data $(\mathbf{e},\mathbf{u})$ of $C^{*}$ where $\mathbf{e} =\{e(t) = M_{r}^{-1}(y)(t) - y(t)\}_{t=0}^{N-1}$. An optimal iFIR controller $C$ that best approximates $C^{*}$ can be tuned via the data-fitting problem in \eqref{eq:obj1DOF_P1}. According to VRFT theory \cite[Sec. 3.1]{bib:Campi2002}, in the case of $P_{1}$ being linear time invariant (LTI), solving \eqref{eq:obj1DOF_P1} is equivalent to solving 
\begin{equation}
    \min_{C(z)} \left \| M_r(z) - \frac{P_{1}(z)C(z)}{1+P_{1}(z)C(z)} \right\|_{2}^{2}
\end{equation}
when the data length is infinite $N \to \infty$,  $\mathbf{u}$ is the realization of a stationary and ergodic stochastic process and $C^{*}(z)$ can be represented by the structure in \eqref{eq:ifir}.  

It is natural to consider the response to disturbances in addition to the reference tracking. Consider the closed loop in Figure \ref{fig:blockdiag1DOF_P1P2}.
$P_{2}$ is a stable, discrete, possibly nonlinear, time-invariant SISO system that represents the disturbance dynamics. $d$ is the disturbance and we assume that it is a known excitable signal during the design phase of our controller. Given two stable discrete reference models chosen by the user, $M_{r}$ and $M_{d}$, for arbitrary $\mathbf{r}$ and $\mathbf{d}$ we aim at tuning the iFIR controller $C$ as the minimizer of $\min || \mathbf{y} - \mathbf{y}^{*}   ||_{2}^{2}$ where $\mathbf{y}^{*} = \{y^{*}(t) = M_{r}(r)(t) + M_{d}(d)(t)\}_{t=0}^{N-1}$. 

\begin{figure}[htbp]
    \begin{center} \vspace{1mm}
    \includegraphics[width=0.74\columnwidth]{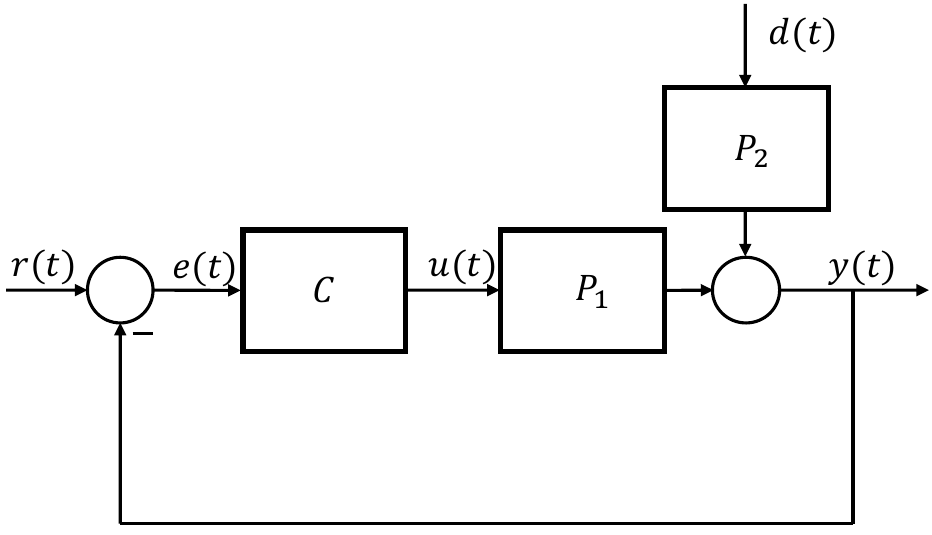} 
        \caption{Closed-loop system with disturbance dynamics.} \vspace{-8mm}
        \label{fig:blockdiag1DOF_P1P2}
    \end{center}
\end{figure}

We address this control problem by extending \eqref{eq:obj1DOF_P1}. Let $\mathbf{u}$ and $\mathbf{d}$ be the exciting signals to $P_{1}$ and $ P_{2}$, respectively. Let $\mathbf{y}$ be the measured plant's output. Assuming perfect matching is achieved, we can derive the data of the virtual reference $\mathbf{r}$ as $r(t) = M_{r}^{-1}(y - M_{d}(d))(t)$. Hence, we obtain the input-output data $(\mathbf{e},\mathbf{u})$ of the ideal controller $C^{*}$ with $\mathbf{e} = \{ e(t) = M_{r}^{-1}(y - M_{d}(d)(t)) - y(t)\}_{t=0}^{N-1}$. Thus, the new data-fitting problem can be formulated as 
\begin{equation} \label{eq:obj1DOF_P1P2}
    \min_{C(z)} \frac{1}{N} \sum_{t=0}^{N-1} \left( u(t) -  C \left( M_{r}^{-1}(y - M_{d}(d))(t)  - y(t) \right)   \right)^{2}.
\end{equation}

Unfortunately, the minimizing controller of \eqref{eq:obj1DOF_P1P2} may not have enough
degrees of freedom to approximate the target models $M_r$ and $M_d$ well.
For instance, for linear plants, the ideal controller achieving perfect matching 
satisfies 
$M_r(z) = \frac{P_{1}(z)C^*(z)}{1+P_{1}(z)C^*(z)}$
and $ M_d(z) = \frac{P_{2}(z)}{1+P_{1}(z)C^*(z)}$, 
where the right-hand side of these identities are the transfer functions from $\mathbf{r}$ to $\mathbf{y}$ and from $\mathbf{d}$ to $\mathbf{y}$ of the closed loop in Figure \ref{fig:blockdiag1DOF_P1P2}. We get,
\begin{equation}
\label{eq:perfect_matching_2dof}
C^{*}(z) = \frac{M_r(z)}{P_{1}(z)(1-M_{r}(z))}
\ \mbox{ and } \
C^{*}(z) = \frac{P_{2}(z) - M_{d}(z)}{M_{d}(z)P_{1}(z)}
\end{equation}
Hence, unless $\frac{M_r(z)}{P_{1}(z)(1-M_{r}(z))} = \frac{P_{2}(z) - M_{d}(z)}{M_{d}(z)P_{1}(z)}$, perfect matching is not possible. This poses significant constraints on the choice of the target models $M_{r}(z)$ and $M_{d}(z)$, \cite[Chp. 2.3.3]{bib:BCE2011}.

Greater flexibility can be achieved by the classical two-degree-of-freedom (2DOF) structure \cite{bib:Campi2002b} in Figure \ref{fig:blockdiag2DOF_P1P2}, 
where $C$ is an iFIR in \eqref{eq:ifir} and $F$ is a FIR prefilter given by
\begin{equation} \label{eq:prefilter}
    F(z) = \sum_{t=0}^{m_{\mathrm{ff}}-1} g^{\mathrm{ff}}(t) z^{-t} \qquad m_{ff} \in \mathbb{N}.
\end{equation}
where $\mathbf{g}^{\mathrm{ff}} = \{g^{\mathrm{ff}}(t)  \}_{t=0}^{m_{\mathrm{ff}}-1}$ are the impulse response coefficients of the FIR prefilter.
\begin{figure}[htbp]
    \begin{center} \vspace{2mm}
        \includegraphics[width=0.86\columnwidth]{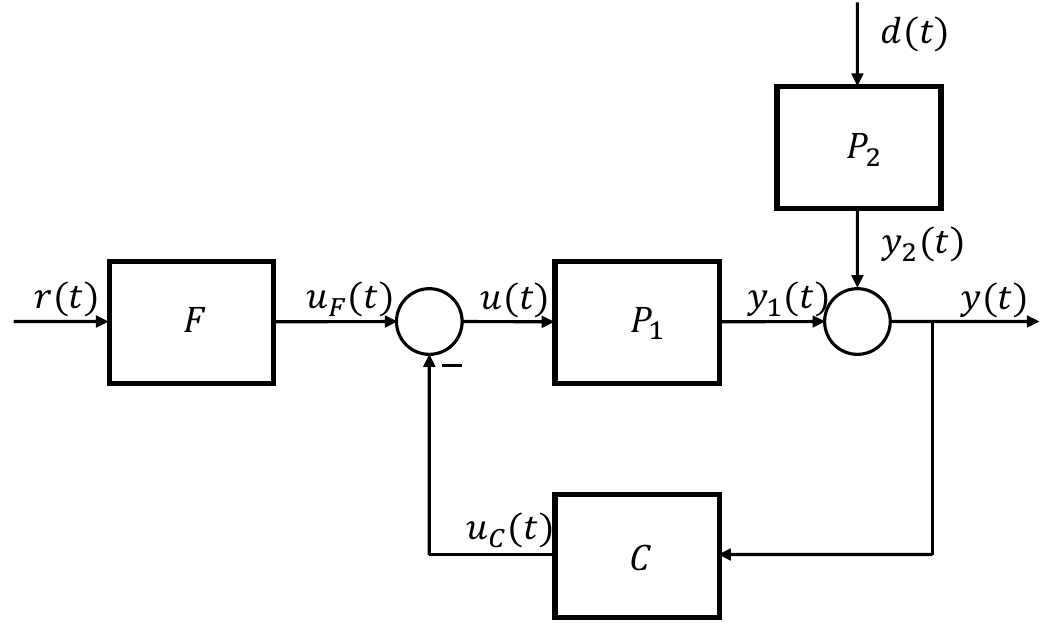} \vspace{-1mm}
        \caption{A block diagram of a closed-loop system considering disturbance dynamics. A 2DOF controller is used.} 
        \label{fig:blockdiag2DOF_P1P2}
    \end{center} \vspace{-6mm}
\end{figure}

For an arbitrary pair $(\mathbf{r}, \mathbf{d})$, we aim at tuning the 2DOF iFIR $(F,C)$ of
Figure \ref{fig:blockdiag2DOF_P1P2} as the minimizer of $\min || \mathbf{y} - \mathbf{y}^{*}   ||_{2}^{2}$ where $\mathbf{y}^{*} = \{y^{*}(t) = M_{r}(r)(t) + M_{d}(d)(t)\}_{t=0}^{N-1}$.
That is,
\begin{equation} \label{eq:obj2DOF_P1P2}
    \min_{F(z),C(z)} \frac{1}{N} \sum_{t=0}^{N-1} \left(  u(t) - F\left(M_{r}^{-1} (y - M_{d}(d) )\right)(t) + C(y)(t)   \right)^{2}
\end{equation}
where $M_{r}^{-1}(z)$ in \eqref{eq:obj1DOF_P1}, \eqref{eq:obj1DOF_P1P2} and \eqref{eq:obj2DOF_P1P2} can be realized by choosing $M_{r}(z)$ as a strictly proper transfer function, or by a causal approximation of $M_{r}(z)$, or by filtering the data in \eqref{eq:obj1DOF_P1}, \eqref{eq:obj1DOF_P1P2} and \eqref{eq:obj2DOF_P1P2} with an additional pre-filter. For example, filtering the data in \eqref{eq:obj2DOF_P1P2} with $M_{r}(z)$ leads to 
\begin{equation} \label{eq:obj2DOF_P1P2_filter}
    \min_{F(z),C(z)} \frac{1}{N} \sum_{t=0}^{N-1} \left(  M_{r}(u)(t) - F(y - M_{d}(d) )(t) + C(M_{r}(y))(t)   \right)^{2}.
\end{equation}

The restriction of $F$ and $C$ to FIR and iFIR controllers, respectively, allows us to
formulate \eqref{eq:obj1DOF_P1}, \eqref{eq:obj1DOF_P1P2}, \eqref{eq:obj2DOF_P1P2}, and \eqref{eq:obj2DOF_P1P2_filter} as least-squares problems in the unknowns $\gamma$, $\mathbf{g}^{\mathrm{fb}}$, and $\mathbf{g}^{\mathrm{ff}}$.
These can be solved efficiently, as shown in \cite[Sec. II]{bib:Wang2024} for \eqref{eq:obj1DOF_P1}. The other cases can be derived in a similar way.
The 2DOF structure of Figure \ref{fig:blockdiag2DOF_P1P2} has enough degrees of
freedom to achieve perfect matching \eqref{eq:perfect_matching_2dof} but is limited
by the restriction to the FIR and iFIR controllers in $F$ and $C$. In the tradition of VRFT design, 
the minimizer of \eqref{eq:obj2DOF_P1P2} should correspond to the minimizer of
\begin{equation} \label{eq:mrc2DOF}
    \min_{F(z), C(z)} \left \| M_r(z) - \frac{P_{1}(z)F(z)}{1\!+\!P_{1}(z)C(z)}  \right\|_{2}^{2} \!\!\!+  \left \| M_d(z) - \frac{P_{2}(z)}{1\!+\!P_{1}(z)C(z)}  \right\|_{2}^{2}\!\!.
\end{equation}
However, since this paper focuses on dissipativity, we do not explore this connection in detail. This will be the subject of future work. The numerical results in Section \ref{sec:gripper} further suggest that this equivalence holds.

In the following, we denote the controller structure in Figure \ref{fig:blockdiag1DOF_P1} and \ref{fig:blockdiag1DOF_P1P2} as 1DOF iFIR and denote the 2DOF structure (prefilter + feedback controller) in Figure \ref{fig:blockdiag2DOF_P1P2} as 2DOF iFIR. We simply say 1DOF/2DOF iFIR when we refer to either 1DOF iFIR or 2DOF iFIR. iFIR controller $C$ always refers to \eqref{eq:ifir} and is
used in both 1DOF iFIR and 2DOF iFIR.

\section{Design of dissipative 1DOF/2DOF iFIR} 
\label{sec:core_contribution}
\subsection{Stability of dissipative systems} \label{sec:dissipativityDefi}
The data-driven designs in \eqref{eq:obj1DOF_P1} and \eqref{eq:obj2DOF_P1P2} do not guarantee closed-loop stability. When we have the prior knowledge of the dissipativity condition of a plant (e.g., through estimation \cite{bib:Koch2022,bib:Martin2023c}), the interconnection theorems of dissipative systems \cite[Chp. 10]{bib:Khalil1996} guarantee closed-loop stability if the feedback controller satisfies a certain dissipativity condition. This implies that we can have a stability-guaranteed data-driven design by solving \eqref{eq:obj1DOF_P1} and \eqref{eq:obj2DOF_P1P2} with dissipativity constraints on the 1DOF/2DOF iFIR. 

First, the concept of dissipativity will be defined in an input-output framework for SISO systems, adapting \cite[Def. 2]{bib:HillANDMoylan1980}, \cite[App. C.2]{bib:GoodwinANDSin1984}. 
We will focus on two classical dissipativity properties, the passivity condition and the gain condition, adapting the notation of \cite[Def. 2.12]{bib:Rodolphe1996}. In what follows, $T_{+} = \{0,1,2,...\}$ is a set of time instants, $\ell_{2}(T_{+})$ is the Hilbert space of sequences $u:T_{+} \to \mathbb{R}$, and $\ell_{2_{e}}(T_{+})$ is the extended $\ell_{2}(T_{+})$.

\begin{definition}
Consider any operator $\mathcal{G}: \ell_{2_{e}}(T_{+}) \to \ell_{2_{e}}(T_{+})$ with input $u$
and output $y$. 

\noindent\textbullet~ $\mathcal{G}$ is $\nu\in\mathbb{R}$ input feedforward passive (IFP) and $\rho\in\mathbb{R}$ output feedback passive (OFP) if 
    \begin{equation} \label{eq:supply_passivity} 
       \sum_{t=0}^{t_f} \begin{bmatrix}
            y(t) \\ u(t)
        \end{bmatrix}^{T} \begin{bmatrix}
            -\rho & \frac{1}{2} \\ \frac{1}{2} & -\nu 
        \end{bmatrix}\begin{bmatrix}
            y(t) \\ u(t)
        \end{bmatrix} \geq 0,
    \end{equation}
    for all $u \in \ell_{2_{e}}(T_{+})$ and for all $t_{f} \in T_{+}$. 
    In particular, we say that $\mathcal{G}$ has an excess/shortage of input passivity for positive/negative $\nu$, and has an excess/shortage of output passivity for positive/negative $\rho$. $\mathcal{G}$ is passive for $\nu = \rho = 0$. 

\noindent\textbullet~$\mathcal{G}$ has a $\ell_{2}$ gain less than or equal to $0 \leq  \alpha\in \mathbb{R}$ if  
    \begin{equation} \label{eq:supply_gain}
       \sum_{t=0}^{t_f} \begin{bmatrix}
            y(t) \\ u(t)
        \end{bmatrix}^{T} \begin{bmatrix}
            -1 & 0 \\ 0 & \alpha^2
        \end{bmatrix}\begin{bmatrix}
            y(t) \\ u(t)
        \end{bmatrix} \geq 0,
    \end{equation}
    for all $u \in \ell_{2_{e}}(T_{+})$ and for all $t_{f} \in T_{+}$.
     $\hfill\lrcorner$
\end{definition}    

In the rest of the paper, we use the $\nu_{1}$, $\rho_{1}$, $\alpha_{1}$ to denote the dissipativity properties of $P_{1}$, 
such that $P_{1}$ is OFP with $\rho_{1}$, IFP with $\nu_{1}$ and has $\ell_{2}$ gain less than or equal to $\alpha_{1}$. Likewise, we use $\nu_{c}$, $\rho_{c}$, $\alpha_{c}$ to denote the dissipativity properties of the iFIR controller $C$ in \eqref{eq:ifir}, such that $C$ is OFP with $\rho_{c}$, IFP with $\nu_{c}$ and has $\ell_{2}$ gain less than or equal to $\alpha_{c}$. 
We will \emph{assume} that
\begin{itemize}
\item both $P_{1}$ and $P_{2}$ have finite $\ell_{2}$ gains;
\item the passivity condition $\nu_{1},\rho_{1}$ or the gain condition $\alpha_{1}$ are known.
\end{itemize}
We will then design the iFIR controller $C$ to achieve compatible dissipative conditions that 
guarantee closed-loop stability. These conditions are detailed in the following
theorem, which emphasizes three cases that have a simple graphical interpretation.

\begin{theorem} \label{thm:dissipativityCases}
    For the closed loops in Figures \ref{fig:blockdiag1DOF_P1} or \ref{fig:blockdiag2DOF_P1P2}, consider the following dissipativity cases for $P_{1}$ 
    \begin{itemize}
        \item[\textbf{A:}]
        $\nu_{1}$ and $\rho_1$ known. $\alpha_{1}$ unknown. 
        $\nu_{1} \leq 0$, $\nu_{1}\rho_{1} < \frac{1}{4}$;
        \item[\textbf{B:}]
        $\nu_{1}$ and $\rho_{1}$ known. $\alpha_{1}$ unknown. \textcolor{black}{ $\nu_{1} > 0$}; 
        \item[\textbf{C:}]
         $\nu_{1}$ and $\rho_{1}$ unknown. $\alpha_{1}$ known. 
    \end{itemize}  
    Then, the closed loop in Figure \ref{fig:blockdiag1DOF_P1} and the closed loop in Figure \ref{fig:blockdiag2DOF_P1P2} \textcolor{black}{
    have a finite $\ell_{2}$ gain from $r$ to $y$}, if the iFIR controller $C$ satisfies, respectively,
    \begin{itemize}
       \item[\textbf{A:}] $\rho_{c} = -\nu_{1} + \epsilon_{1}, \nu_{c} = -\rho_{1} + \epsilon_{2}$;     
       \item[\textbf{B:}]
        $\rho_{c} = 0, \nu_{c} = -\rho_{1} + \epsilon_{2}$; 
        \item[\textbf{C:}] 
        $\alpha_{c} = \frac{1}{\alpha_{1}}-\epsilon_{3}$. 
    \end{itemize}
    where $\epsilon_{1},\epsilon_{2},\epsilon_{3}$ are (small) positive constants.   $\hfill\lrcorner$
\end{theorem}

\begin{proof}
    We focus on the closed loop of Figure \ref{fig:blockdiag1DOF_P1}.
    
    For cases $\textbf{A}$ and $\textbf{B}$ we have 
    \begin{equation} \label{eq:NuRho}
            \nu_{c} + \rho_{1} >0 \qquad
            \nu_{1} + \rho_{c} > 0.
    \end{equation}
    By the passivity theorem {\cite[Thm. 10.6]{bib:Khalil1996} }, \eqref{eq:NuRho} guarantees that the closed loop \textcolor{black}{
    has a finite $\ell_{2}$ gain from $r$ to $y$}.
    
    For case $\textbf{C}$ we have 
    \begin{equation} \label{eq:NuRho_smallgain}
        \alpha_{1}\alpha_{c} < 1.
    \end{equation}
    By the small-gain theorem {\cite[Thm. 10.5]{bib:Khalil1996} }, \eqref{eq:NuRho} guarantees that the closed loop \textcolor{black}{
    has a finite $\ell_{2}$ gain from $r$ to $y$}.

    For the closed loop in Figure \ref{fig:blockdiag2DOF_P1P2}, 
    both \eqref{eq:NuRho} and \eqref{eq:NuRho_smallgain} imply that there is a finite $\ell_{2}$ gain from $(u_{F},y_{2})$ to $(y_{1},u_{C})$.  Since $F$ has a finite $\ell_{2}$ gain by construction, and $P_{2}$ has a finite $\ell_{2}$ gain by assumption, there is a finite $\ell_{2}$ gain from $(r,d)$ to $(y_{1},u_{C})$. Since $y = y_{1} + y_{2}$, there is a finite $\ell_{2}$ gain from $(r,d)$ to $y$.
\end{proof}

All these cases have a convenient graphical interpretation \cite[Sec. 3]{bib:Willems1972}, \cite[Lemma. C.3.2]{bib:GoodwinANDSin1984}, illustrated in Figure \ref{fig:nyplots_passivity}.
The Nyquist diagram of the iFIR controller $C$
 \begin{itemize}
    \item[\textbf{A:}] lies inside the closed
    disk centred at $\frac{1}{2\rho_{c}}$ with radius $\frac{1}{2\rho_{c}} \sqrt{1-4\nu_{c} \rho_{c} }$;
    \item[\textbf{B:}] lies in the infinite set of the complex plane whose elements
    have real part greater than or equal to $\nu_{c}$;
    \item[\textbf{C:}] lies inside the disk with radius $\alpha_c$ centred at zero.  
\end{itemize}

\begin{figure}[htbp]
    \begin{center}
        \includegraphics[width=0.96\columnwidth]{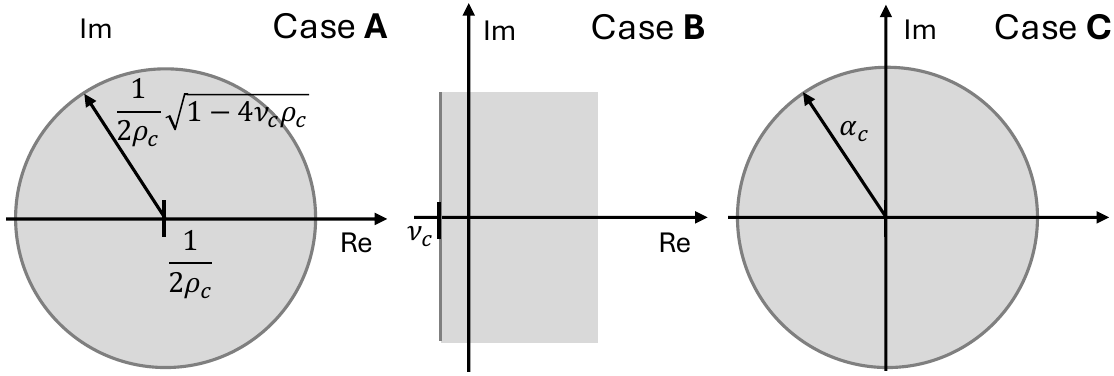}
        \caption{Graphical interpretation for Cases \textbf{A}-\textbf{C} in Theorem \ref{thm:dissipativityCases}. The shaded areas represent the regions where
        the Nyquist diagram of iFIR controller $C$ must lie.}
        \label{fig:nyplots_passivity} 
    \end{center} \vspace{-4mm}
\end{figure}

Note that if $\nu_{1}\rho_{1} < \frac{1}{4}$ of Case \textbf{A} of Theorem \ref{thm:dissipativityCases} is violated, then the radius $\frac{1}{2\rho_{c}} \sqrt{1-4\nu_{c} \rho_{c} }$ is not defined (imaginary number). This provides a justification for the condition
$\nu_{1}\rho_{1} < \frac{1}{4}$ of Case \textbf{A}.

\textcolor{black}{
Cases \textbf{A} and \textbf{C} are feasible only if the integrator is removed by the iFIR ($\gamma = 0$), as they both require finite $\ell_2$ gain.
} For simplicity of the exposition, in what follows we will not insist on this distinction.

\subsection{Dissipativity constraints for 1DOF/2DOF iFIR} \label{sec:dissi_constraints}
In this section, we illustrate how to combine the least-squares optimization 
of Section \ref{sec:vrft} with dissipativity constraints on the iFIR controller $C$. The goal is to guarantee that the controller $C$ is the best dissipative controller that fits the data and
satisfies the conditions of Theorem \ref{thm:dissipativityCases}. 

The first formulation makes use of the KYP lemma and is based on the following
state-space realization of FIR part of iFIR controller $C$
\begin{equation}  \label{eq:ifirStatespace}
    \begin{cases}
        \quad x(t+1) = A_{c}x(t) + B_{c}w(t) \\
        \quad q(t) = C_{c}x(t) + D_{c}w(t)
    \end{cases}
\end{equation}  
 where $w(t),q(t)$ are its input and output at the time instant $t$. $A_{c}\in \mathbb{R}^{m_{\mathrm{fb}-1} \times m_{\mathrm{fb}-1}}, B_{c}  \in \mathbb{R}^{m_{\mathrm{fb}-1}} ,C_{c} \in \mathbb{R}^{1 \times m_{\mathrm{fb}-1}},D_{c} \in \mathbb{R}$
are given by
\begin{align} \label{eq:ifirStatespace_matrices}
    A_{c} &= \begin{bmatrix}
        0 & 1 & 0 &0 &\dots &0 & 0\\
        0 & 0 & 1 &0 &\dots &0 & 0\\
        \vdots & & & & & \\
        0 & 0 & 0 &0 &\dots &1 & 0 \\
        0 & 0 & 0 &\dots & &0 & 1 \\
        0 & 0 & 0 &\dots & &0 & 0 
    \end{bmatrix} , \quad  
    B_{c} = \begin{bmatrix}
        0\\
        0  \\
        \vdots \\
        0 \\ 
        0\\
        1
    \end{bmatrix} \nonumber \\
    C_{c} &= 
    \begin{bmatrix}
        g^{\mathrm{fb}}(m_{\mathrm{fb}}-1) & g^{\mathrm{fb}}(m_{\mathrm{fb}}-2) & \dots & g^{\mathrm{fb}}(1) 
    \end{bmatrix} \nonumber \\
    D_{c} &= g^{\mathrm{fb}}(0) \, . 
\end{align}
With this representation, the dissipativity constraints correspond
to the following LMIs.

\begin{theorem} \label{thm:kyp}
Under the assumptions of Theorem \ref{thm:dissipativityCases}, 
iFIR controller $C$ whose FIR part is represented by \eqref{eq:ifirStatespace},\eqref{eq:ifirStatespace_matrices} 
satisfies \textbf{A}, \textbf{B}, \textbf{C}
in Theorem \ref{thm:dissipativityCases} 
 if there exist $  0 < X=X^T \in \mathbb{R}^{m_{\mathrm{fb}-1} \times m_{\mathrm{fb}-1}}$,
$C_c$, and $D_c$ such that, respectively,
\begin{subequations}
\label{eq:kyp}
\begin{align}
\mbox{\textbf{A:}} &
        \begin{bmatrix}
        \!A_{c}^{T}XA_{c}- X & A_{c}^{T}XB_{c} - \frac{C_{c}^T}{2} & C_{c}^{T}(\varepsilon_1\!-\!\nu_1)^{\frac{1}{2}} \\
        \!B_{c}^{T}XA_{c} \!-\! \frac{C_{c}}{2} \!\!&\! \!B_{c}^{T}XB_{c}\!-\!D_{c} \!\!+\!\! (\varepsilon_2\!-\!\rho_1) \!\!&\!\! D_{c}(\varepsilon_1\!-\!\nu_1)^{\frac{1}{2}} \\
        \!(\varepsilon_1\!-\!\nu_1)^{\frac{1}{2}}C_{c} & (\varepsilon_1\!-\!\nu_1)^{\frac{1}{2}}D_{c} &  - I
        \end{bmatrix} \!\!<\! 0,  \nonumber \\
        & \gamma = 0; \label{eq:kypA} \\
\mbox{\textbf{B:}} &
        \begin{bmatrix}
        A_{c}^{T}XA_{c}- X & A_{c}^{T}XB_{c} - \frac{C_{c}^T}{2}  \\
        B_{c}^{T}XA_{c} - \frac{C_{c}}{2} & B_{c}^{T}XB_{c}-D_{c} + (\varepsilon_2-\rho_1) 
        \end{bmatrix} < 0 \,,\ \gamma \geq 0; \label{eq:kypB} \\
\mbox{\textbf{C:}} &
       \begin{bmatrix}
        A_{c}^{T}XA_{c}- X & A_{c}^{T}XB_{c}  & C_{c}^{T} \\
        B_{c}^{T}XA_{c}  & B_{c}^{T}XB_{c} - (\frac{1}{\alpha_{1}}-\epsilon_{3})^{2} & D_{c} \\
       C_{c} & D_{c} &  - I
        \end{bmatrix} < 0 \,,\ \gamma = 0. \label{eq:kypC}
\end{align} \vspace{-5mm}
\end{subequations}

$\hfill\lrcorner$
\end{theorem}

Combining the constraints of Theorem \ref{thm:kyp} with \eqref{eq:obj1DOF_P1} or \eqref{eq:obj2DOF_P1P2}
leads to a convex constrained optimization problem that can be
solved using standard solvers such as CVXPY \cite{bib:cvxpy}. 
The minimizer is a dissipative 1DOF/2DOF iFIR that best solves the data-fit problem in \eqref{eq:obj1DOF_P1}, \eqref{eq:obj2DOF_P1P2}. 
However, as shown in \cite[Sec. IV]{bib:Wang2024}, constrained optimization based on LMIs of the form \eqref{eq:kyp} has scalability issues. These LMIs are not computationally tractable for high-order iFIR controllers, as the coefficients of the unknown matrix $X$ grow quadratically with the order of the controller. This motivates the development of the following dissipativity constraints in the frequency domain.

\begin{theorem} \label{thm:sample}
Given the iFIR controller $C$ represented by \eqref{eq:ifir},
define 
\begin{subequations}
\label{eq:f_expressions}
\begin{align}
f_{r}(\theta) &= \sum_{t=0}^{m_{\mathrm{fb}}-1}g^{\mathrm{fb}}(t)\cos \left(t\theta\right) , \quad \forall \theta \in [0,\pi] \label{eq:fr} \\
f_{i}(\theta) &= \sum_{t=0}^{m_{\mathrm{fb}}-1}g^{\mathrm{fb}}(t)\sin \left(t\theta\right) , \quad \forall \theta \in [0,\pi]. \label{eq:fi}
\end{align}
\end{subequations}
For any selection of the \emph{sampling parameter} $2\leq M \in \mathbb{N}$
and of the \emph{decay rate parameters} $h_0>0$ and $0<h \leq 1$,
 define
\begin{equation}
    \epsilon = \pi h_{0}\frac{1-h^{m_{\mathrm{fb}}}}{1-h}\frac{m_{\mathrm{fb}}-1}{2M}. \label{eq:epsilon}
\end{equation}

Under the assumptions of Theorem \ref{thm:dissipativityCases}, 
the iFIR controller $C$ satisfies \textbf{A}, \textbf{B}, \textbf{C}
in Theorem \ref{thm:dissipativityCases} if
\begin{equation}
|g^{\mathrm{fb}}(t)| \leq h_0 h^t  \quad \forall t \in \{0,1,...,m_{\mathrm{fb}}-1\} \label{eq:sampleDecay}\\
\end{equation}
and, respectively, \vspace{2mm}

\noindent\textbf{A:} \vspace{-11mm}
\begin{subequations}
\label{eq:freq_conditions_A}
    
        \begin{align}
        \gamma &= 0 \label{eq:sampleGamma}\\
          a_{1} \!+\! \epsilon  &< f_{r}\left(\theta \!=\! \frac{m\pi}{M}\right) < a_{2} \!-\! \epsilon, \ \forall m \!\in\! \{0,1,...,M\} \label{eq:sample_fr}\\
          \frac{-r}{\sqrt{2}} \!+\! \epsilon  &< f_{i}\left(\theta \!=\! \frac{m\pi}{M}\right) < \frac{r}{\sqrt{2}} \!-\!\epsilon, \ \forall m \!\in\! \{0,1,...,M\}  \label{eq:sample_fi}
        \end{align}
    where
    \begin{align}
        a_{k} &= \frac{1}{2(-\nu_{1}+\epsilon_{1})} + \frac{(-1)^k }{\sqrt{2}}r \label{eq:aCaseA}\\
        r &=\frac{1}{2(-\nu_{1}+\epsilon_{1})}\sqrt{1-4(-\nu_{1}+\epsilon_{1})(-\rho_{1}+\epsilon_{2})};
        \label{eq:rCaseA}
    \end{align}
\end{subequations} \vspace{2mm}

\noindent\textbf{B:} \vspace{-7mm}
    \begin{subequations} \label{eq:sampleRHP}
    \label{eq:freq_conditions_B} 
        \begin{align}
            \gamma &\geq 0 \label{eq:sampleRHPgamma}\\
                f_{r}\left(\theta = \frac{m\pi}{M}\right) &> (\epsilon_{2}- \rho_{1}) + \epsilon, \quad \forall m \in \{0,1,...,M\}; \label{eq:sampleRHP_fr}
        \end{align} 
    \end{subequations} \vspace{2mm}
    
\noindent\textbf{C:}  \vspace{-7mm}
    \begin{subequations}
    \label{eq:freq_conditions_C} 
        \begin{align}
            \gamma &= 0 \label{eq:sampled_gamma2} \\
              -r + \epsilon  &< f_{r}\left(\theta = \frac{m\pi}{M}\right) <r  - \epsilon, 
              \quad \forall m \in \{0,1,...,M\} \label{eq:sampled_gain_real}\\
              -r+ \epsilon  &< f_{i}\left(\theta = \frac{m\pi}{M}\right) < r -\epsilon, \quad \forall m \in \{0,1,...,M\} 
              \label{eq:sampled_gain_imag}
        \end{align}
    where
    \begin{align}
    r = \frac{1}{\sqrt{2}}\left(\frac{1}{\alpha_1}-\varepsilon_3\right). 
    \end{align}
    \end{subequations}
\end{theorem}
\vspace{2mm}

Although the conditions of Theorem \ref{thm:sample} may appear complicated and not intuitive
at first sight, the conditions \eqref{eq:freq_conditions_A}, \eqref{eq:freq_conditions_B}, and \eqref{eq:freq_conditions_C} characterize three sets of linear inequalities in the coefficients $\mathbf{g}^{\mathrm{fb}}$ of the iFIR controller. The number of these inequalities
is proportional to the size of the controller $m_{\mathrm{fb}}$ and the sampling parameter $M$,
which guarantees scalability when paired to \eqref{eq:obj1DOF_P1} or \eqref{eq:obj2DOF_P1P2}.

The complexity of Theorem \ref{thm:sample} is also only apparent. Looking at the frequency domain, the linear inequalities \eqref{eq:freq_conditions_A} guarantee that the Nyquist plot of the iFIR controller is confined within the blue box of Figure \ref{fig:box}. 
Given the graphical interpretation of passivity and small gain properties in Figure \ref{fig:nyplots_passivity}, it should be intuitively clear why these inequalities
provide sufficient conditions for Theorem \ref{thm:dissipativityCases}. The interpretation of
\eqref{eq:freq_conditions_B} and \eqref{eq:freq_conditions_C} is similar.

\begin{figure}[htbp]
    \begin{center}
    \vspace{2mm}
        \includegraphics[width=0.56\columnwidth]{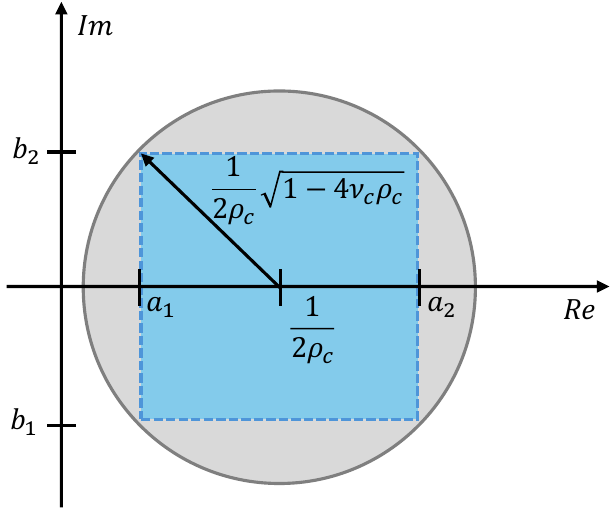}
        \caption{Equation \eqref{eq:freq_conditions_A} constrains the Nyquist diagram of the controller $C$ within the blue box, to satisfy the passivity property given by the grey disk. 
        } 
        \label{fig:box}
    \end{center}
        \vspace{-5mm}
\end{figure}

The box constraint in Figure \ref{fig:box} is clearly a conservative approximation of the disk. This can be mitigated using more complex polygons, such as an octagon. More complex polygons can be characterized through a large set of linear inequalities (whose number remains proportional
to $m_{\mathrm{fb}}$ and the sampling parameter $M$). This will be the subject of future work. 

Theorem \ref{thm:sample} constrains the iFIR controller $C$ by enforcing
inequalities on $M$ samples of its Nyquist locus. This is why all these conditions have an additional correction factor $\epsilon$, defined in \eqref{eq:epsilon}. The magnitude of 
$\epsilon$ is related to the variation of the Nyquist
locus between samples and guarantees that the constrained Nyquist samples are representative of the whole Nyquis locus. In fact, $\epsilon$ gets smaller as the number of samples $M$ increases.
In general, \eqref{eq:epsilon} is a conservative estimation and numerical examples show
that a much smaller $\epsilon$ is often enough to guarantee the desired property of the controller. 

Finally, we observe that the VRFT design \eqref{eq:obj1DOF_P1} and \eqref{eq:obj2DOF_P1P2}
combined with the dissipativity constraints in Theorem \ref{thm:kyp} or Theorem \ref{thm:sample} leads
to a data-driven design approach with stability guarantees. \textcolor{black}{
Noisy measures, insufficient data and nonlinearities} could affect the quality of the achieved performance but cannot compromise the stability of the closed loop, which is guaranteed by passivity or small gain theorems. The cost of such a strong guarantee is that the dissipativity constraints proposed in this section do impose conservativeness on the achievable performance of 1DOF/2DOF iFIR if the ideal controller does not satisfy the dissipativity constraints. However, this trade-off between performance (the objective functions in \eqref{eq:obj1DOF_P1}, \eqref{eq:obj2DOF_P1P2}) and stability (Theorem \ref{thm:kyp} or Theorem \ref{thm:sample}) is optimally handled, via convex optimization.

\subsection{Proofs} \label{sec:proof}
\subsubsection*{Proof of Theorem \ref{thm:kyp}}
We show that \eqref{eq:kypA}, \eqref{eq:kypB}, \eqref{eq:kypC} are equivalent to the standard KYP lemma in {\cite[Chp. 4.6.2]{bib:CaverlyANDForbes2024}}, and hence the iFIR controller $C$ satisfies \eqref{eq:supply_passivity} or \eqref{eq:supply_gain} with $\rho_{c},\nu_{c},\alpha_{c}$ specified in Theorem \ref{thm:dissipativityCases}. 

Let \eqref{eq:kypA} be partitioned as 
\begin{equation} \label{eq:partition}
    \begin{bmatrix}
        \bar{A} & \bar{B}\\
        \bar{B}^T & \bar{C} 
\end{bmatrix} > 0
\end{equation}
where 
$$
\bar{A} = -\begin{bmatrix}
        \!A_{c}^{T}XA_{c}- X & A_{c}^{T}XB_{c} - \frac{C_{c}^T}{2}\\
        \!B_{c}^{T}XA_{c} \!-\! \frac{C_{c}}{2} & B_{c}^{T}XB_{c}\!-\!D_{c} \!\!+\!\! (\varepsilon_2\!-\!\rho_1) 
\end{bmatrix} 
$$
$$
\bar{B} = -\begin{bmatrix}
        C_{c}^{T}(\varepsilon_1\!-\!\nu_1)^{\frac{1}{2}}\\
        D_{c}(\varepsilon_1\!-\!\nu_1)^{\frac{1}{2}}
\end{bmatrix} 
\ \mbox{ and } \
\bar{C} = I. 
$$
Then by Schur complement, \eqref{eq:kypA} is equivalent to {
\begin{equation}
    \bar{C} > 0
    \ \mbox{ and } \
    \bar{A}- \bar{B}\bar{C}^{-1}\bar{B}^T > 0 \, .\label{eq:schur}
\end{equation} }
The second inequality is equivalent to the inequality of the KYP lemma in {\cite[Chp. 4.6.2]{bib:CaverlyANDForbes2024} (readers are also referred to \cite[Lemma. 3]{bib:Kottenstette2014} and \cite[Lemma. C.4.2]{bib:GoodwinANDSin1984})}, if we set 
the matrices $Q,S,R,P$ in  {\cite[Chp. 4.6.2]{bib:CaverlyANDForbes2024} } as $Q = - (\epsilon_{1} - \nu_{1}), S = \frac{1}{2}, R = -( \epsilon_{2} - \rho_{1}), P = X$.
Therefore, we get
\begin{equation} \label{eq:supply_caseA} 
       \sum_{t=0}^{t_f} \begin{bmatrix}
            q(t) \\ w(t)
        \end{bmatrix}^{T} \begin{bmatrix}
            \nu_{1}- \epsilon_{1}  & \frac{1}{2} \\ \frac{1}{2} &  \rho_{1}-\epsilon_{2}
            
        \end{bmatrix}\begin{bmatrix}
            q(t) \\ w(t)
        \end{bmatrix} \geq 0,
    \end{equation}
for all $w \in \ell_{2_{e}}(T_{+})$ and for all $t_{f} \in T_{+}$ \cite{bib:Soh1999}. That is, given $\gamma = 0$, the iFIR controller $C$ satisfies \eqref{eq:supply_passivity} with $\rho_{c} = - \nu_{1}+ \epsilon_{1} , \nu_{c} = -\rho_{1}+ \epsilon_{2}$, as required.

Equation \eqref{eq:kypB} is equivalent to the
inequality of the KYP lemma in {\cite[Chp. 4.6.2]{bib:CaverlyANDForbes2024}} if we set the matrices $Q,S,R,P$ in {\cite[Chp. 4.6.2]{bib:CaverlyANDForbes2024} }as $Q = 0, S = \frac{1}{2}, R = -( \epsilon_{2} - \rho_{1}), P = X$. Hence, since $\gamma \geq 0$, the iFIR controller $C$ satisfies \eqref{eq:supply_passivity} with $\rho_{c} = 0 , \nu_{c} = -\rho_{1}+ \epsilon_{2}$. 

For\eqref{eq:kypC}, consider the partition
\eqref{eq:partition} for 
$$
\bar{A} = -\begin{bmatrix}
        \!A_{c}^{T}XA_{c}- X & A_{c}^{T}XB_{c}\\
        \!B_{c}^{T}XA_{c}  & B_{c}^{T}XB_{c} - (\frac{1}{\alpha_{1}} - \epsilon_{3})^{2} 
\end{bmatrix} 
$$
$$
\bar{B} = -\begin{bmatrix}
        C_{c}^{T}\\
        D_{c}
\end{bmatrix} 
\ \mbox{ and } \
\bar{C} = I .
$$
By Schur complement, the second inequality 
in \eqref{eq:schur} is equivalent to the inequality of the KYP lemma in {\cite[Chp. 4.6.2]{bib:CaverlyANDForbes2024}}, if we set $Q,S,R,P$ in  {\cite[Chp. 4.6.2]{bib:CaverlyANDForbes2024}} as $Q = -1, S = 0, R = (\frac{1}{\alpha_{1}} - \epsilon_{3})^{2}, P = X$. Hence, for $\gamma = 0$, the iFIR $C$ satisfies  \eqref{eq:supply_gain} with $\alpha_{c} = \frac{1}{\alpha_{1}}  - \epsilon_{3}$. 
\hfill$\blacksquare$

\subsubsection*{Proof of Theorem \ref{thm:sample}, Case \textbf{A}}

For \eqref{eq:sampleGamma}, the 
frequency representation of the iFIR controller is
$C(e^{j\theta}) = \sum_{t=0}^{m_{\mathrm{fb}}-1}g^{\mathrm{fb}}(t)e^{-j\theta t}$. 
For any given set of parameters of $a_{1},a_{2},b_{1},b_{2} \in \mathbb{R}$ such that $a_{2} > a_{1}, b_{2} > b_{1}$, and for all $\theta \in [0,\pi]$, we have
\begin{subequations}
\begin{align}
    a_{1} &< \mathrm{Re}[C(e^{j\theta})] < a_{2} \label{eq:eqva_1} \\
    &\iff 2a_{1} < C(e^{j\theta}) + C(e^{-j\theta}) < 2a_{2} \\
    &\iff a_{1} <  f_{r}(\theta)  < a_{2} \label{eq:eqva_2}\\
    \frac{-r}{\sqrt{2}} &< \mathrm{Im}[C(e^{j\theta})] < \frac{r}{\sqrt{2}} \label{eq:eqvb_1} \\
    &\iff \frac{-2r}{\sqrt{2}} <  -j(C(e^{j\theta}) - C(e^{-j\theta}))< \frac{2r}{\sqrt{2}} \\
    &\iff \frac{-r}{\sqrt{2}} <  f_{i}(\theta) < \frac{r}{\sqrt{2}}  \label{eq:eqvb_2}
\end{align}
\end{subequations}
where $f_{r},f_{i}$ are defined in \eqref{eq:fr}, \eqref{eq:fi}. 

Thus, \eqref{eq:eqva_2} and \eqref{eq:eqvb_2}
guarantee that the Nyquist diagram of the iFIR controller $C$ belongs to the interior of a rectangular box region. From \eqref{eq:aCaseA}, \eqref{eq:rCaseA}, taking $c = \frac{1}{2(-\nu_{1}+\epsilon_{1})}$,
\eqref{eq:eqva_1} and \eqref{eq:eqvb_1} define an open box whose four vertices lie on the boundary of the disk centred at $c$ with radius $r$.
Thus, by construction, as shown in Figure \ref{fig:box}, the Nyquist plot of the iFIR controller C satisfies
\eqref{eq:supply_passivity} for $\rho = \epsilon_1 -\nu_1$ and $\nu = \epsilon_2 -\rho_1$.

We now show that \eqref{eq:sample_fr},\eqref{eq:sample_fi} combined to \eqref{eq:epsilon},\eqref{eq:sampleDecay} imply \eqref{eq:eqva_2} and \eqref{eq:eqvb_2}.
Note that \eqref{eq:sample_fr} and \eqref{eq:sample_fi} sample
\eqref{eq:eqva_2} and \eqref{eq:eqvb_2} at $\theta = \frac{m\pi}{M} \in [0,\pi]$ for $0 \leq m \leq M$.

\cite[Thm. 4]{bib:Wang2024} shows that the decay constraint \eqref{eq:sampleDecay} guarantees the following, for all $M \geq 2, 0<h\leq1, h_{0} > 0$,
\begin{equation}
    \left| f_{r}(\theta) -  f_{r}\left(\theta = \frac{m\pi}{M}\right) \right| \leq \epsilon, \quad \forall m \in \{0,1,...,M\} 
\end{equation}
where $\epsilon$ is from \eqref{eq:epsilon}. Hence, \eqref{eq:epsilon}, \eqref{eq:sampleDecay}, \eqref{eq:sample_fr} imply \eqref{eq:eqva_2}.

A similar result holds for $f_i$.
For all $\theta \in [0,\pi]$ and $\Delta \in \mathbb{R}$, 
\begin{align}
    |f_{i}(\theta + \Delta) - f_{i}(\theta)| &\leq \sum_{t=0}^{m_{\mathrm{fb}}-1} |g^{\mathrm{fb}}(t)||\sin(t\theta + t\Delta) - \sin(t\theta)| \nonumber \\
    &\leq \sum_{t=0}^{m_{\mathrm{fb}}-1} |g^{\mathrm{fb}}(t)||t\Delta| \label{eq:random1} \nonumber \\
    &\leq (m_{\mathrm{fb}}-1) |\Delta| \sum_{t=0}^{m_{\mathrm{fb}}-1} |g^{\mathrm{fb}}(t)|.
\end{align}
Hence, following the derivation of \cite[Eq. 17]{bib:Wang2024}, we get
\begin{equation}
    \left| f_{i}(\theta) -  f_{i}\left(\theta = \frac{m\pi}{M}\right) \right| \leq \epsilon, \quad \forall m \in \{0,1,...,M\}.
\end{equation}
for $\epsilon$ given in \eqref{eq:epsilon}. Hence, \eqref{eq:epsilon}, \eqref{eq:sampleDecay}, \eqref{eq:sample_fi} imply \eqref{eq:eqvb_2}.

\subsubsection*{\textcolor{black}{
Case \textbf{B}}}
We need to show that, for all $\theta \in [0,\pi]$,
\begin{equation}
    Re[C(e^{j\theta})]  \geq \epsilon_{2}- \rho_{1}
\end{equation}
For the iFIR controller, this is satisfied if 
the integrator has a non-negative gain \eqref{eq:sampleRHPgamma} 
and the real part of the Nyquist diagram of the FIR part of iFIR controller is greater or equal to $-\rho_1 + \epsilon_2$. 
The latter is guaranteed by 
\begin{equation}
    \ f_{r}\left(\theta\right) > \epsilon_{2}- \rho_{1} .\label{eq:sampleRHP_fr_cont}
    \end{equation}
\eqref{eq:sampleRHP_fr_cont}
is equivalent to \eqref{eq:eqva_2} for $a_{2} = +\infty$ and $a_{1} =- \rho_{1} + \epsilon_{2}$. Thus, following the approach of the later part of the proof of Theorem \ref{thm:sample} Case \textbf{A}, \eqref{eq:sampleRHP_fr_cont}
is implied by \eqref{eq:epsilon}, \eqref{eq:sampleDecay}, \eqref{eq:sampleRHP_fr}. 

\subsubsection*{\textcolor{black}{
Case \textbf{C}}}
Given \eqref{eq:sampled_gamma2}, the
frequency representation of the iFIR controller reduces to
$C(e^{j\theta}) = \sum_{t=0}^{m_{\mathrm{fb}}-1}g^{\mathrm{fb}}(t)e^{-j\theta t}$. For 
$ -a_1 = -b_1 = a_2 = b_2 =\frac{1}{\sqrt{2}}\left(\frac{1}{\alpha_1}-\epsilon_2\right)$,
 \eqref{eq:eqva_1} and \eqref{eq:eqvb_1} define an open box
 whose four vertices are on the boundary of a disk of
  radius $\frac{1}{\alpha_1}-\epsilon_2$, centered at $0$. 
 Hence, the Nyquist plot of the iFIR controller C satisfies
\eqref{eq:supply_gain} for $\alpha = \frac{1}{\alpha_1} -\epsilon_2$.
 
From the proof of Case \textbf{A}, \eqref{eq:eqva_2} and \eqref{eq:eqvb_2} are implied by \eqref{eq:sampleDecay}, \eqref{eq:epsilon}, \eqref{eq:sampled_gain_real} and by \eqref{eq:sampleDecay}, \eqref{eq:epsilon}, \eqref{eq:sampled_gain_imag}, respectively. $\hfill\blacksquare$

\section{Examples} \label{sec:examples} 
\subsection{Gripper reference tracking and impedance shaping}\label{sec:gripper}
Figure \ref{fig:gripper} represents a simplified model of a
two-finger gripper, whose fingers are covered by a soft material.
The control goal is to design the 2DOF controller of Figure \ref{fig:blockdiag2DOF_P1P2} to shape the gripper's reference tracking and disturbance dynamics. $y$ represents the opening displacement of the gripper ($y = 0$ for the closed gripper) measured between the rigid part of the fingers. Likewise, $u$ is the opening control force acting on the rigid part of the fingers. The gripper is passive from $u$ to $\dot{y}$. So, the input to the iFIR controller $C$ is the velocity $\dot{y}$.

\begin{figure}[htbp]
    \begin{center}
        \includegraphics[width=0.8\columnwidth]{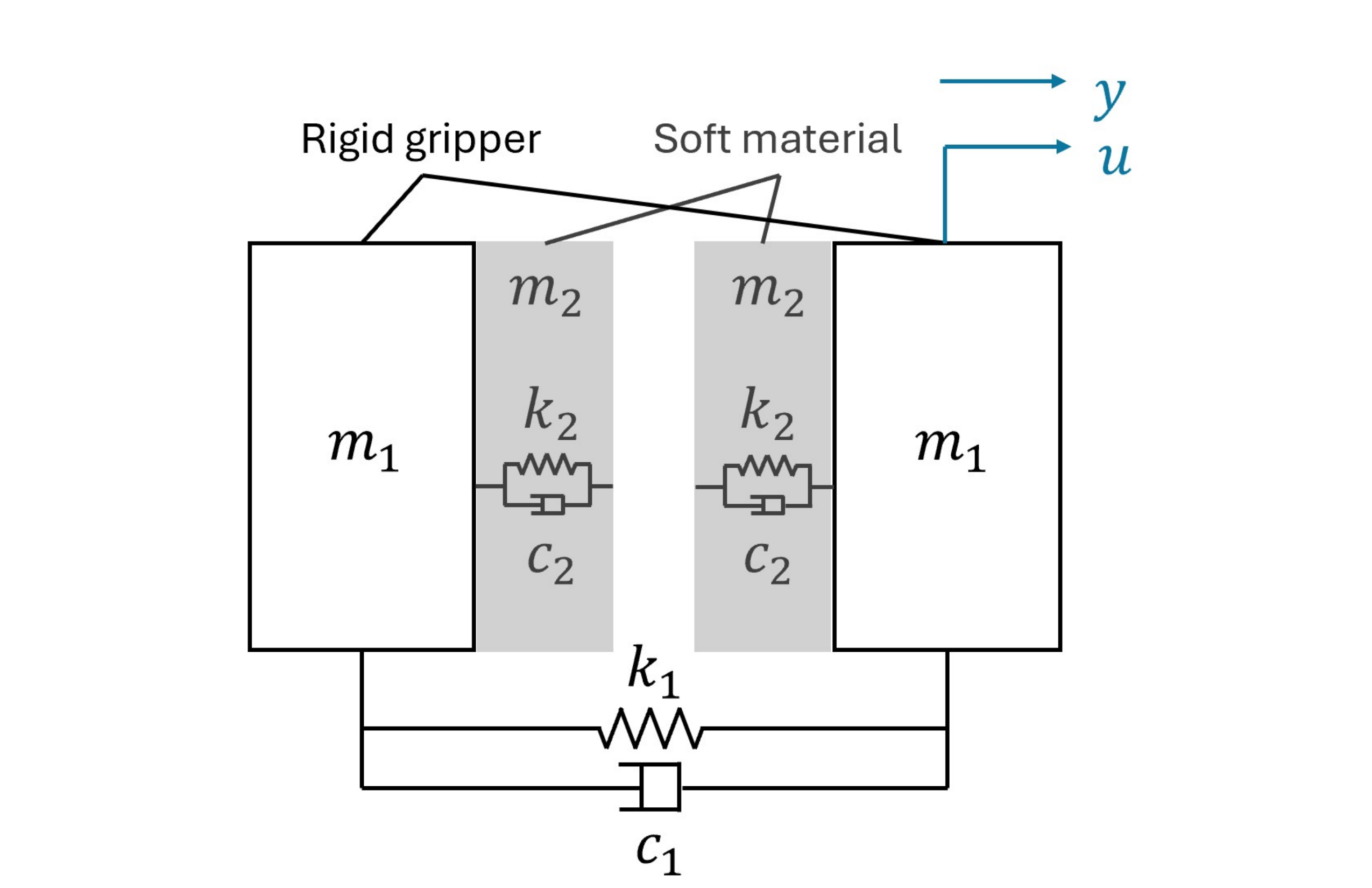} \vspace{-3mm}
        \caption{\textcolor{black}{A compliant two finger gripper. Its rigid structure is modeled by the masses $m_{1}$ and the compliant part is modelled by $m_{2}$ through spring/damper elements. The gripper closes driven by a mechanical component and opens through the action of a motor $u$.
        }} 
        \label{fig:gripper}
    \end{center} \vspace{-5mm}
\end{figure}

Rigid and soft parts of the gripper have masses $m_1=0.01$ kg, $m_2=0.005$ kg, respectively. The compliance of the soft material is modeled by a spring with stiffness $k_2=1$ N/m and a damper with damping parameter $c_2=0.2$ Ns/m. The gripper is mechanically constrained to close automatically through a spring
with stiffness $k_1=1.5$ N/m and \textcolor{black}{ a damper with damper parameter $c_1=0.1$ Ns/m.}

The control goals are defined by the reference models for tracking
and disturbance dynamics
\begin{equation*}
    M_r(s) = \frac{150}{(s+10)(s+15)}
    \quad 
    M_d(s) = \frac{1000}{(s+5)(s+10)(s+30)},
\end{equation*}
whose Bode diagrams are in Figure \ref{fig:bode_gripper_ref}.
The disturbance dynamics is chosen to preserve the open-loop dynamics
at low frequency but correct the resonance at high frequency. That is, we do not alter the mechanical impedance of the gripper at low frequencies, but we remove high-frequency oscillations. The reference tracking dynamics guarantees good tracking in the low frequency regime and a significant roll-off at high frequency, removing resonant oscillations. \vspace{-4mm}
\begin{figure}[htbp]
    \begin{center}
        \includegraphics[width=0.49\columnwidth]{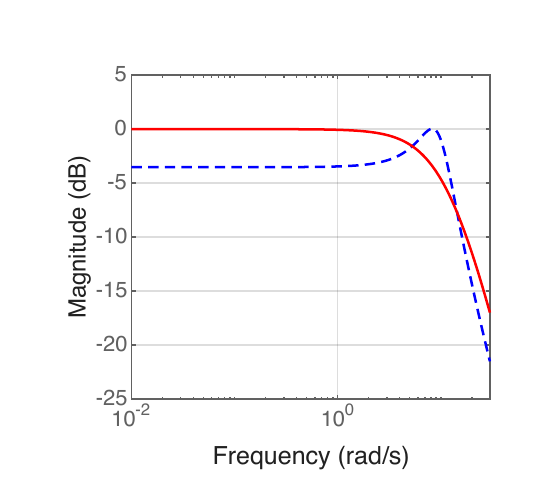} 
        \includegraphics[width=0.49\columnwidth]{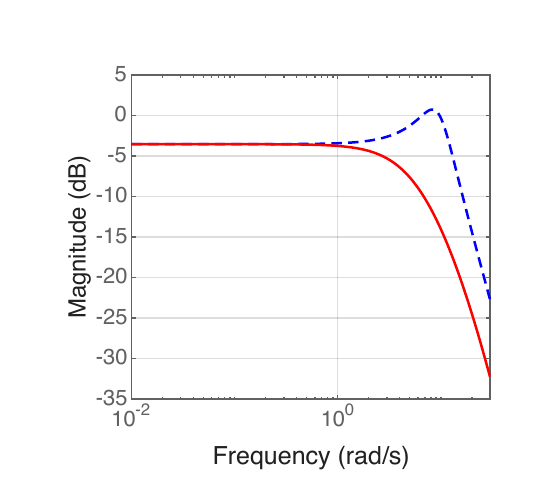} \vspace{-5mm}
        \caption{Bode diagrams for reference tracking (left) and disturbance response (right). Red: target dynamics. Blue-dashed: open loop.}  
        \label{fig:bode_gripper_ref}
    \end{center} \vspace{-3mm}
\end{figure}

For the controller design, we discretized the plant with zero order hold at a sampling time $T_{s} = 0.01s$. This introduces a shortage of input passivity in the plant. The plant is excited in open loop with $u(t) = \sum_{i=1}^{10} \cos(\omega_{i}T_{s}t + \phi^u_i)$ and $d(t) = \sum_{i=1}^{10} \cos(\omega_{i}T_{s}t + \phi^d_i)$ where $\omega_{i}$ linearly spans the frequency range of $[0.5,10]$ rad/s, 
$0 \leq \phi^u_i, \phi^d_i \leq \pi$,
and $t \in \{0,...,N-1\}$ for $N = 10000$. Several controllers are derived using \eqref{eq:obj2DOF_P1P2_filter} with and without the constraints of Section \ref{sec:dissi_constraints}, namely (i) optimal unconstrained PD controller, (ii) optimal unconstrained FIR controller ($m_{\mathrm{fb}}=50$), and (iii) optimal constrained FIR ($m_{\mathrm{fb}}=50$) based on Case A of Theorem \ref{thm:sample}. PD and unconstrained FIR are designed with clean data. The constrained FIR controller has been designed with clean and noisy data. 

The Nyquist diagrams of the FIR controllers are shown in Figure \ref{fig:nyquist_gripper}. 
All controllers are tested in simulation. 
Initially, the reference position is set at $0.02$m.
Later, a step disturbance of $0.005$N acting on $m_2$ 
is added. The time plots are shown in Figure
\ref{fig:step_response_gripper}. These show
a comparison between constrained FIR controllers and the ideal response of the target model. The comparison includes the FIR controller trained with outputs corrupted by a zero mean white noise, resulting in an average signal-to-noise ratio of $28.1$ dB for position and $30.6$ dB for velocity. Both FIR controllers are stable and show an excellent match with the target model, as further illustrated by the closed-loop Bode diagrams in Figure \ref{fig:bode_results_gripper}. The presence of noise in the training data slightly compromised the performance of the FIR controller.
In contrast, the optimal unconstrained FIR and the optimal PD controller ($C(z)=0.3979 +\frac{0.0136T_s}{1-z^{-1}}$) are unstable.
\vspace{-2mm}

\begin{figure}[htbp]
    \begin{center}
        \includegraphics[width=0.45\columnwidth]{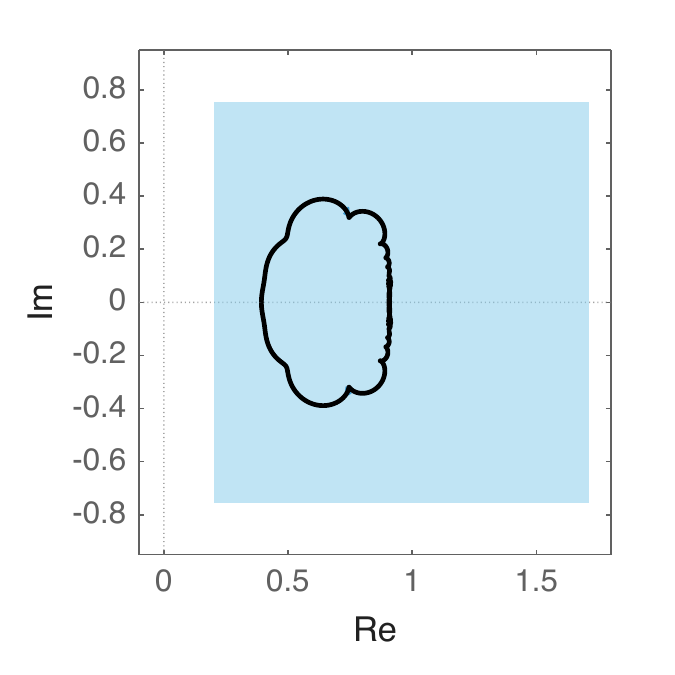}
        \includegraphics[width=0.45\columnwidth]{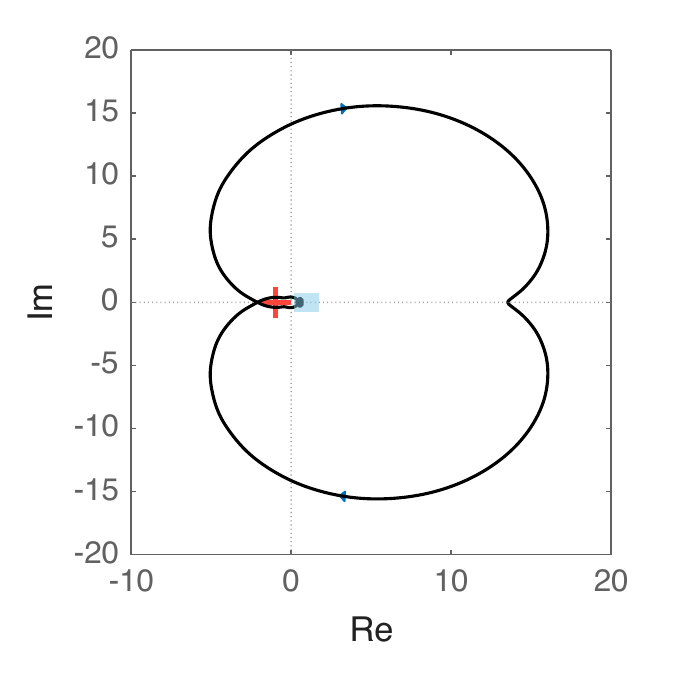} \vspace{-2mm}
        \caption{Nyquist diagrams of $C$. Left: constrained FIR (\eqref{eq:obj2DOF_P1P2_filter}, Theorem \ref{thm:sample}, clean data). Right: FIR without box constraint (\eqref{eq:obj2DOF_P1P2_filter}, clean data).} 
        \label{fig:nyquist_gripper}
    \end{center} \vspace{-4mm}
\end{figure}
\begin{figure}[htbp]
    \begin{center}
        \includegraphics[width=0.49\columnwidth]{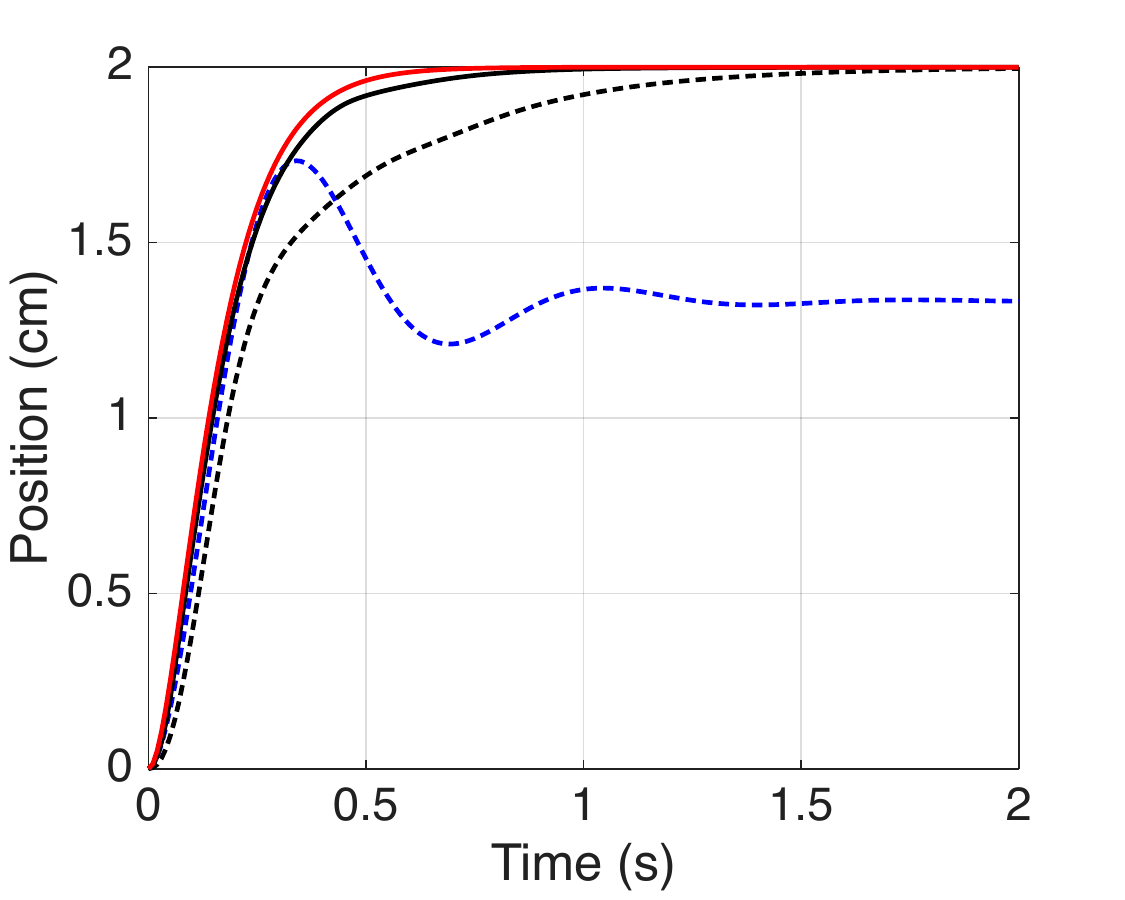}
        \includegraphics[width=0.49\columnwidth]{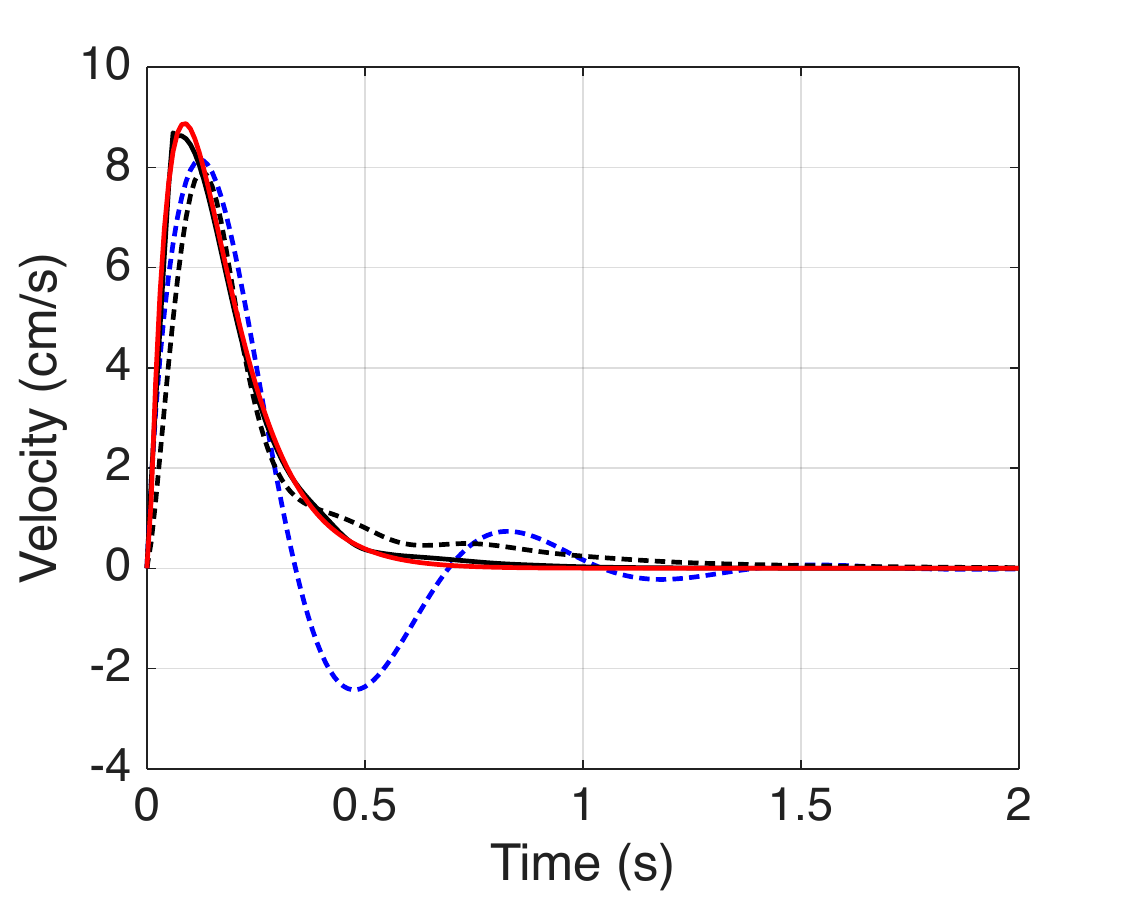}
        \includegraphics[width=0.49\columnwidth]{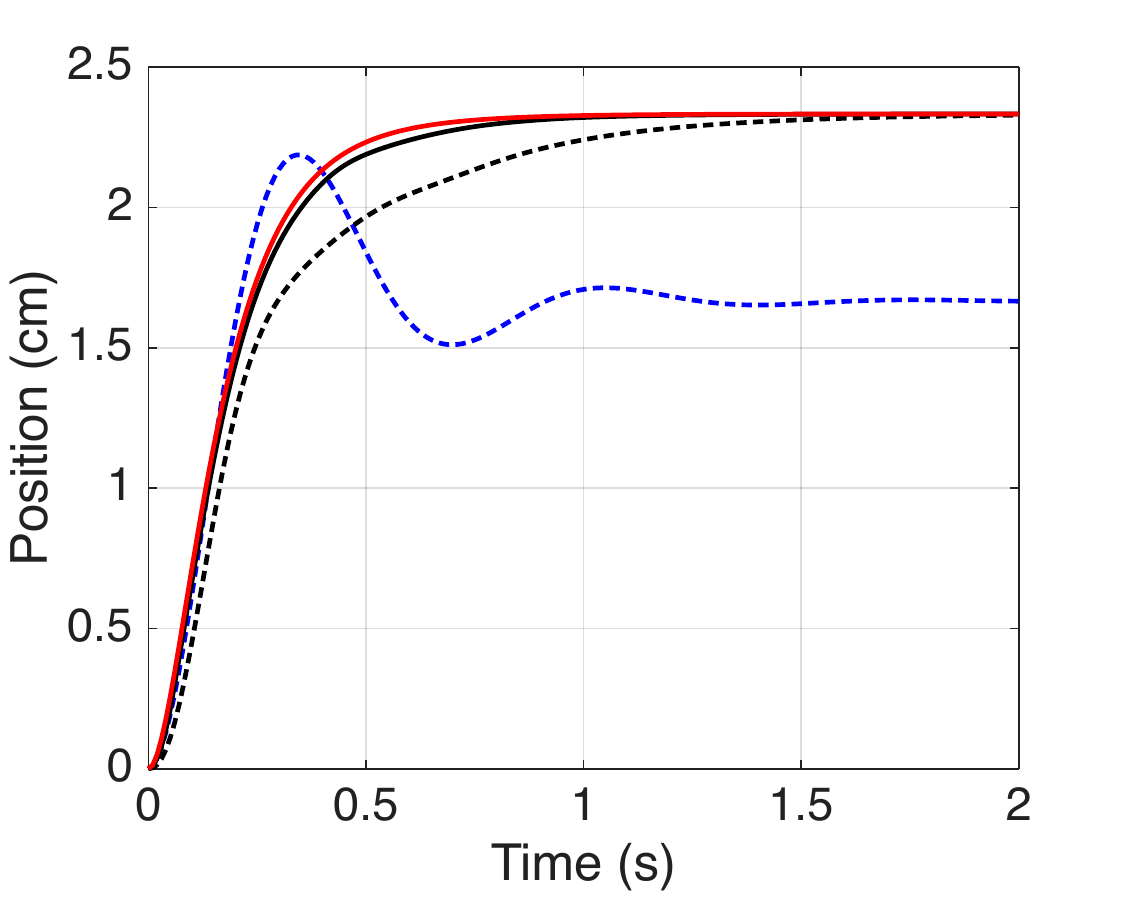}
        \includegraphics[width=0.49\columnwidth]{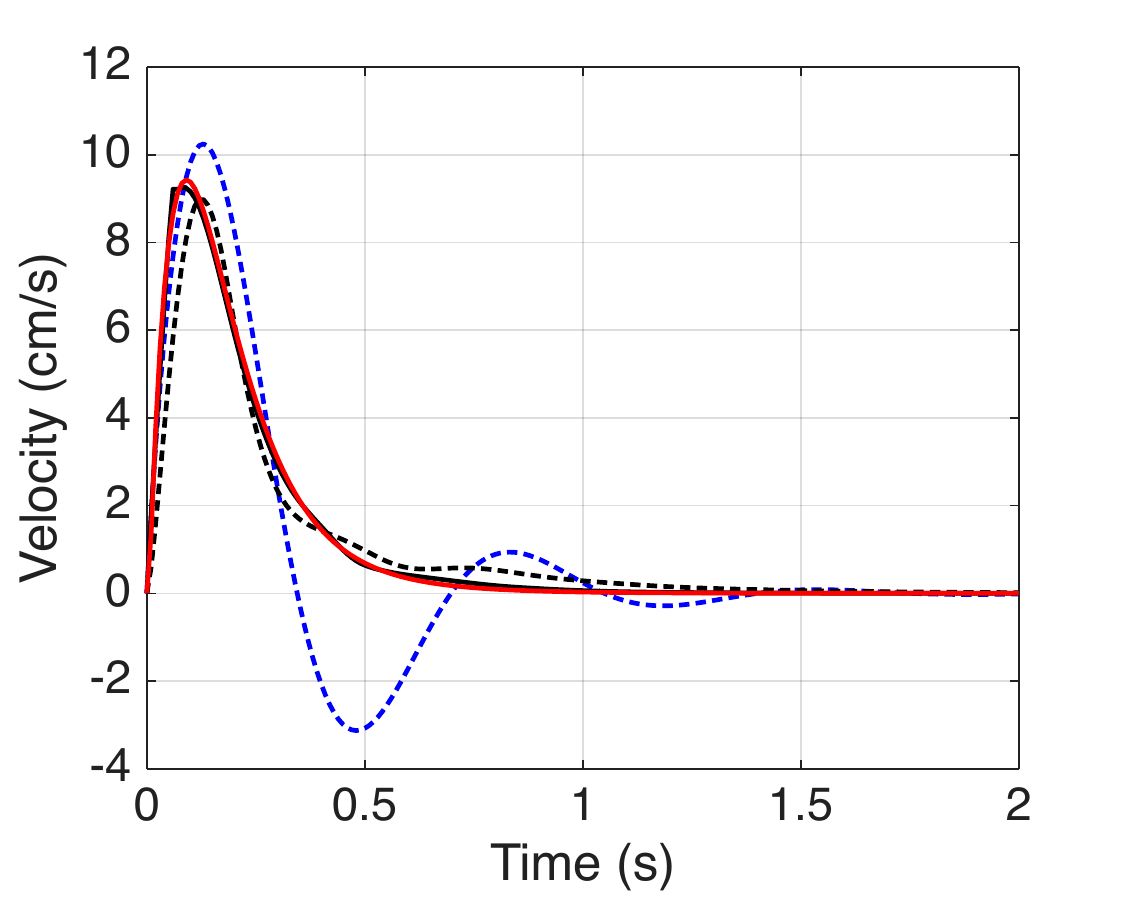}
        \vspace{-4mm}
        \caption{Reference tracking (top row) and reference tracking with a step disturbance (bottom row). Red: target response; black-continuous: closed-loop response with FIR controller trained with clean data; black-dashed: closed-loop response with FIR controller trained with noisy data; blue-dashed: open loop.} 
        \label{fig:step_response_gripper}
    \end{center} \vspace{-3mm}
\end{figure}

\begin{figure}[htbp]
\vspace{-3mm}
    \begin{center}
        \includegraphics[width=0.49\columnwidth]{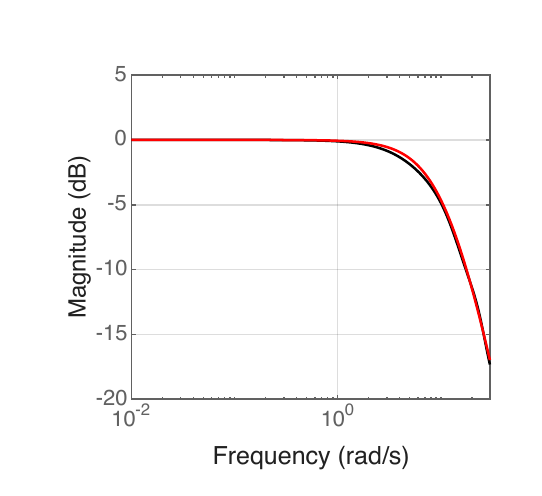}
        \includegraphics[width=0.49\columnwidth] {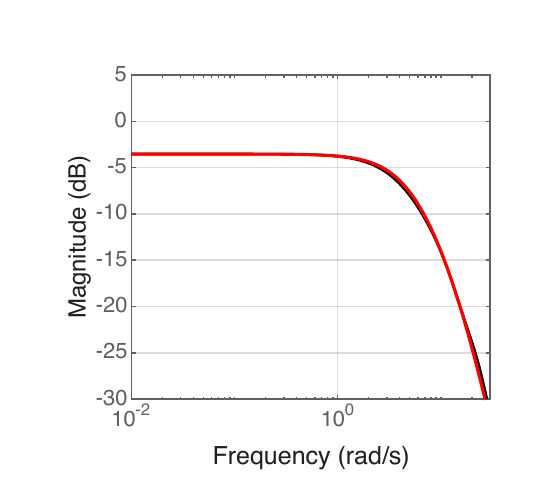}
        \vspace{-7mm}
        \caption{Bode diagrams for reference tracking (left) and disturbance response (right). Red: target reference model. Black-continuous: closed-loop transfer function with FIR controller (clean data).} 
        \label{fig:bode_results_gripper}
    \end{center} \vspace{-4mm}
\end{figure}

The numerical benchmark of Table \ref{tb:KYP_BOX} clarifies the
difference, in terms of computational complexity, between 
the constraints of Theorems \ref{thm:kyp} and \ref{thm:sample}.
The LMIs of Theorem \ref{thm:kyp} are more computationally expensive
than the frequency domain constraints of Theorem \ref{thm:sample}.
\textcolor{black}{ The results of Table \ref{tb:KYP_BOX} are obtained using CVXPY \cite{bib:cvxpy} with MOSEK solver \cite{bib:mosek}, on a standard laptop with Apple Silicon M2 chip. Compilation time is not included.}  

\begin{table}[htbp]
    \begin{center}
        \begin{tabular}{c|cccc} 
                 &  $m_{\mathrm{fb}}=25$ & $m_{\mathrm{fb}}=50$ & $m_{\mathrm{fb}}=100$ & $m_{\mathrm{fb}}=200$   \\ \hline
            KYP &  0.242 s&  0.554 s & 7.731 s &  186.428 s  \\ \hline
            Box, $M=300$ & 0.191 s& 0.221 s & 0.346 s & 0.569 s  \\
            Box, $M=500$ & 0.176 s& 0.239 s & 0.353 s & 0.590 s  \\
            Box, $M=700$ & 0.192 s& 0.258 s & 0.376 s & 0.704 s  \\
        \end{tabular} 
         \caption{Computation time: LMIs vs frequency constraints.} \label{tb:KYP_BOX} 
    \end{center} \vspace{-10mm}
\end{table}

\section{Conclusion}
The combination of virtual reference feedback tuning and dissipativity theory leads to a constrained least-squares problem where the least-squares objective function captures the performance target (reference tracking and disturbance model shaping), and the dissipativity constraints guarantee closed-loop stability. The formal results are proven and illustrated through a numerical example. A simple numerical benchmark shows the strong reduction in computation time that the new linear constraints offer in comparison to classical LMIs. As a side contribution, we extended VRFT synthesis to disturbance dynamics. Future works will focus on 
\textcolor{black}{
extending our result to nonlinear and to multiple-input multiple-output systems. 
}

\bibliographystyle{IEEEtran}   
\bibliography{IEEEabrv,bib}

\end{document}